\documentclass[11pt]{llncs}
\usepackage{amsmath,amssymb}
\usepackage{graphicx}
\newcommand{\comm}[1]{}

\usepackage{fullpage}
\usepackage[top=0.8in, bottom=0.9in, left=0.9in, right=0.9in]{geometry}
\usepackage[linesnumbered, vlined, ruled]{algorithm2e}

\usepackage{color,soul}
\usepackage{amsmath,amssymb}	
\usepackage[euler]{textgreek}
\usepackage{algpseudocode}
\usepackage{cite}
\usepackage{enumitem}
\graphicspath{{./Fig/}}

\usepackage[pdftex, plainpages = false, pdfpagelabels,
                 bookmarks=false,
                 bookmarksopen = true,
                 bookmarksnumbered = true,
                 breaklinks = true,
                 linktocpage,
                 pagebackref,
                 colorlinks = true,  
                 linkcolor = blue,
                 urlcolor  = cyan,
                 citecolor = red,
                 anchorcolor = green,
                 hyperindex = true,
                 hyperfigures
                 ]{hyperref}

\def\calL{\mathcal{L}}
\def\calH{\mathcal{H}}

\def\opt{S_{\text{opt}}}
\def\optvalue{\delta_{\text{opt}}}
\def\findmin{FindMinCost}
\def\reset{ResetCost}
\newtheorem{observation}{Observation}
\newtheorem{invariant}{Algorithm Invariant}

\begin{document}

\title{On Line-Separable Weighted Unit-Disk Coverage and Related Problems\thanks{A preliminary version of this paper will appear in {\em Proceedings of the 49th International Symposium on Mathematical Foundations of Computer Science (MFCS 2024)}. This research was supported in part by NSF under Grant CCF-2300356.}}
\author{Gang Liu 
\and
Haitao Wang 
}

 \institute{
 Kahlert School of Computing\\
  University of Utah, Salt Lake City, UT 84112, USA\\
  \email{u0866264@utah.edu, haitao.wang@utah.edu}
}

\maketitle

\pagestyle{plain}
\pagenumbering{arabic}
\setcounter{page}{1}

\begin{abstract}
Given a set $P$ of $n$ points and a set $S$ of $n$ weighted disks in the plane, the disk coverage problem is to compute a subset of disks of smallest total weight such that the union of the disks in the subset covers all points of $P$. The problem is NP-hard. In this paper, we consider a line-separable unit-disk version of the problem where all disks have the same radius and their centers are separated from the points of $P$ by a line $\ell$. We present an $O(n^{3/2}\log^2 n)$ time algorithm for the problem. This improves the previously best work of $O(n^2\log n)$ time. Our result leads to an algorithm of $O(n^{{7}/{2}}\log^2 n)$ time for the halfplane coverage problem (i.e., using $n$ weighted halfplanes to cover $n$ points), an improvement over the previous $O(n^4\log n)$ time solution. If all halfplanes are lower ones, our algorithm runs in $O(n^{{3}/{2}}\log^2 n)$ time, while the previous best algorithm takes $O(n^2\log n)$ time. Using duality, the hitting set problems under the same settings can be solved with similar time complexities.
\end{abstract}


\section{Introduction}
\label{sec:intro}

Let $P$ be a set of points, and $S$ a set of disks in the plane such that each disk has a positive weight. The {\em disk coverage} problem asks for a subset of disks whose union covers all points and the total weight of the disks in the subset is minimized. The problem is NP-hard, even if all disks have the same radius and all disks have the same weight~\cite{ref:FederOp88,ref:MustafaIm10}. Polynomial-time approximation algorithms have been proposed for the problem and many of its variants, e.g.,~\cite{ref:AgarwalNe20,ref:BusPr18,ref:ChanEx14,ref:ChanFa20,ref:GanjugunteGe11,ref:LiA15}.

In this paper, we consider a {\em line-separable unit-disk} version of the problem where all disks have the same radius and their centers are separated from the points of $P$ by a line $\ell$ (see Fig.~\ref{fig:unitcase}). 
This version of the problem has been studied before. For the {\em unweighted case}, that is, all disks have the same weight, 
Amb\"uhl, Erlebach, Mihal\'ak, and Nunkesser~\cite{ref:AmbuhlCo06} first solved the problem in $O(m^2n)$ time, where $n=|P|$ and $m=|S|$.
An improved $O(nm+ n\log n )$ time algorithm was later given in \cite{ref:ClaudeAn10}. Liu and Wang~\cite{ref:LiuOn23} recently presented an $O((n+m)\log(n+m)) $ time algorithm.\footnote{The runtime of the algorithm in the conference paper of \cite{ref:LiuOn23} was $m^{2/3}n^{2/3}2^{O(\log^*(m+n))}+O((n+m)\log (n+m))$, which has been improved to $O((n+m)\log (n+m))$ time in the latest arXiv version.}
For the weighted case, Pederson and Wang~\cite{ref:PedersenAl22} derived an algorithm of $O(nm\log (m+n))$ time or $O((m+n)\log(m+n) + \kappa\log m)$ time, where $\kappa$ is the number of pairs of disks that intersect and $\kappa=O(m^2)$ in the worst case.
In this paper, we propose an algorithm of $O(n\sqrt{m}\log^2 m +(m+n)\log(m+n))$ time for the weighted case. In addition to the improvement over the previous work, perhaps theoretically more interesting is that the runtime of our algorithm is subquadratic.

\begin{figure}[h]
\begin{minipage}[t]{\textwidth}
\begin{center}
\includegraphics[height=1.0in]{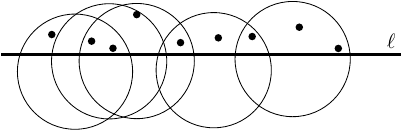}
\caption{\footnotesize Illustrating the line-separable unit-disk case: All points of $P$ are above $\ell$ while the centers of all disks are below $\ell$.}
\label{fig:unitcase}
\end{center}
\end{minipage}
\vspace{-0.15in}
\end{figure}


\paragraph{\bf The halfplane coverage problem.} If every disk of $S$ is a halfplane, then the problem becomes the {\em halfplane coverage problem}. To solve the problem, Pedersen and Wang~\cite{ref:PedersenAl22} showed that the problem can be reduced to $O(n^2)$ instances of the {\em lower-only halfplane coverage problem} in which all halfplanes are lower halfplanes; this reduction works for both the unweighted and the weighted cases. Consequently, if the lower-only problem can be solved in $O(T)$ time, then the general problem (i.e., $S$ has both lower and upper halfplanes) is solvable in $O(n^2\cdot T)$ time.

For the weighted lower-only problem,  Chan and Grant~\cite{ref:ChanEx14} first gave an algorithm that runs in $O((m+n)^4)$ time. As observed in \cite{ref:PedersenAl22}, the lower-only halfplane coverage problem is actually a special case of the line-separable unit-disk coverage problem. Indeed, let $\ell$ be a horizontal line below all the points of $P$. Then, since each halfplane of $S$ is a lower halfplane, it can be considered a disk of infinite radius with center below $\ell$. In this way, the lower-only halfplane coverage problem becomes an instance of the line-separable unit-disk coverage problem. As such, with their algorithm for the weighted line-separable unit-disk coverage problem, Pederson and Wang~\cite{ref:PedersenAl22} solved the weighted lower-only halfplane coverage problem in $O(nm+ n\log n)$ time or $O((m+n)\log(m+n) + m^2\log m)$ time. Using our new algorithm for the weighted line-separable unit-disk coverage problem, the weighted lower-only halfplane coverage problem can now be solved in $O(n\sqrt{m}\log^2 m +(m+n)\log(m+n))$ time.

The unweighted lower-only halfplane coverage problem can be solved faster. Indeed, since the problem is a special case of the unweighted line-separable unit-disk coverage problem, with the $O((n+m)\log (n+m))$ time algorithm of Liu and Wang~\cite{ref:LiuOn23} for the latter problem, the unweighted lower-only halfplane coverage problem is solvable in $O((n+m)\log (n+m))$ time. Wang and Xue~\cite{ref:WangAl24} derived another $O((n+m)\log (n+m))$ time algorithm for the unweighted lower-only halfplane coverage problem with a different approach (which does not work for the unit-disk problem); they also solved the general unweighted halfplane coverage problem in $O(n^{4/3}\log^{5/3}n\log^{O(1)}\log n)$ time. In addition, by a reduction from the set equality problem~\cite{ref:Ben-OrLo83}, a lower bound of $\Omega((n+m)\log (n+m))$ is proved in \cite{ref:WangAl24} for the lower-only halfplane coverage problem under the algebraic decision tree model.

\paragraph{\bf The hitting set problem.} A related problem is the hitting set problem in which each point of $P$ has a positive weight and we seek to find a subset of points with minimum total weight so that each disk of $S$ contains at least one point of the subset.
Since the disks of $S$ have the same radius, the problem is actually ``dual'' to the disk coverage problem.
More specifically, if we consider the set of unit disks centered at the points of $P$ as a set of ``dual disks'' and consider the centers of the disks of $S$ as a set of ``dual points'', then the hitting set problem on $P$ and $S$ is equivalent to finding a minimum weight subset of dual disks whose union covers all dual points. Consequently, applying our new weighted case line-separable unit-disk coverage algorithm in this paper solves the weighted line-separable unit-disk hitting set problem in $O(m\sqrt{n}\log^2 n +(m+n)\log(m+n))$ time; applying the $O((m+n)\log(m+n))$ time algorithm in \cite{ref:LiuOn23} for the unweighted line-separable unit-disk coverage algorithm solves the unweighted line-separable unit-disk hitting set problem in $O((m+n)\log(m+n))$ time.

If every disk of $S$ is a halfplane, then the problem becomes the {\em halfplane hitting set} problem.  Har-Peled and Lee \cite{ref:Har-PeledWe12} first solved the weighted problem in $O((m+n)^6)$ time. Liu and Wang \cite{ref:LiuGe23} showed that the problem can be reduced to $O(n^2)$ instances of the {\em lower-only halfplane hitting set problem} in which all halfplanes are lower halfplanes; this reduction works for both the unweighted and the weighted cases. Consequently, if the lower-only problem can be solved in $O(T)$ time, then the general problem can be solved in $O(n^2\cdot T)$ time. For the lower-only problem, as in the coverage problem, it is a special case of the line-separable unit-disk hitting set problem; consequently,  the weighted and unweighted cases can be solved in $O(m\sqrt{n}\log^2 n +(m+n)\log(m+n))$ time using our new algorithm in this paper and $O((m+n)\log(m+n))$ time using the algorithm in \cite{ref:LiuOn23}, respectively.

\paragraph{\bf Other related work.}
Pedersen and Wang~\cite{ref:PedersenAl22} actually considered a {\em line-constrained} disk coverage problem, where disk centers are on the $x$ -axis while the points of $P$ can be anywhere in the plane, but the disks may have different radii. They solved the weighted case in $O(nm+ n\log n)$ time or $O((m+n)\log(m+n) + \kappa\log m)$ time, where $\kappa$ is the number of pairs of disks that intersect. For the unweighted case, Liu and Wang~\cite{ref:LiuOn23} gave an algorithm of $O((n + m) \log(m + n) + m \log m \log n)$ time.
The line-constrained disk hitting set problem was also studied by Liu and Wang~\cite{ref:LiuGe23}, where an $O((m+n)\log(m+n) + \kappa\log m)$ time algorithm was derived for the weighted case, matching the time complexity of the above line-constrained disk coverage problem. Other types of line-constrained problems have also been considered in the literature, e.g.,\cite{ref:AltMi06,ref:BiloGe05,ref:BiniazFa18,ref:Lev-TovPo05,ref:KarmakarSo13,ref:PedersenOn18,ref:WangLi16}.

\paragraph{\bf Our approach.}
Our algorithm for the weighted line-separable unit-disk coverage problem is essentially a dynamic program. The algorithm description is quite simple and elegant. However, it is not straightforward to prove its correctness. To this end, we show that our algorithm is consistent with the algorithm in \cite{ref:PedersenAl22} for the same problem; one may view our algorithm as a different implementation of the algorithm in \cite{ref:PedersenAl22}. Another challenge of our approach lies in its implementation. More specifically, our algorithm has two key operations, and the efficiency of the algorithm hinges on how to perform these operations efficiently. For this, we construct a data structure based on building a cutting on the disks of $S$~\cite{ref:ChazelleCu93}. Although we do not have a good upper bound on the runtime of each individual operation of the algorithm, we manage to bound the total time of all operations in the algorithm by $O(n\sqrt{m}\log^2 m +(m+n)\log(m+n))$. 

\paragraph{\bf Outlines.} The rest of the paper is structured as follows. After introducing some notation in Section~\ref{sec:pre}, we describe our algorithm and prove its correctness in Section~\ref{sec:algo}. The implementation of the algorithm is presented in Section~\ref{sec:implementation}. 

\section{Preliminaries}
\label{sec:pre}
We follow the notation defined in Section~\ref{sec:intro}, e.g., $P$, $S$, $m$, $n$, $\ell$. All disks of $S$ have the same radius, which we call {\em unit disks}. 
Without loss of generality, we assume that $\ell$ is the $x$-axis and all points of $P$ are above $\ell$ while all centers of disks of $S$ are below $\ell$. Note that when we say that a point is above (or below) $\ell$, we allow the case where the point is on $\ell$. 


We assume that each point of $P$ is covered by at least one disk since otherwise there would be no solution.
Our algorithm can check whether this assumption is met.
For ease of discussion, we make a general position assumption that no point of $P$ lies on the boundary of a disk and no two points of $P$ have the same $x$-coordinate. 



We call a subset $S'$ of $S$ a {\em feasible subset} if the union of all disks of $S'$ covers all points of $P$.
If $S'$ is a feasible subset of minimum total weight, then $S'$ is called an {\em optimal subset}. 
Let $\optvalue$ denote the total weight of all disks in an optimal subset; we call $\optvalue$ the {\em optimal objective value}. 



For any point $q$ in the plane, let $S_q\subseteq S$ denote the subset of disks containing $q$; define $\overline{S_q}=S\setminus S_q$. 
For each disk $s\in S$, let $w(s)$ denote its weight. 

\section{Algorithm description and correctness}
\label{sec:algo}
We now present our algorithm. As mentioned above, the algorithm description is quite simple. The challenging part is to prove its correctness and implement it efficiently. In the following, we first describe the algorithm in Section~\ref{sec:description}, and then establish its correctness in Section~\ref{sec:correct}. The algorithm implementation will be elaborated in Section~\ref{sec:implementation}.  

\subsection{Algorithm description}
\label{sec:description}

We first sort the points of $P$ from left to right as $p_1,p_2,\ldots, p_n$.
Our algorithm then processes the points of $P$ in this order. For each point $p_i\in P$, the algorithm computes a value $\delta_i$. The algorithm also maintains a value $cost(s)$ for each disk $s\in S$, which is initialized to its weight $w(s)$. The pseudocode of the algorithm is given in Algorithm~\ref{algo:10}. 

\vspace{-0.2in}
\begin{algorithm}[h]
	\caption{}
	\label{algo:10}
	\SetAlgoNoLine
	\KwIn{The points of $P$ are sorted from left to right as $p_1,p_2,\ldots,p_n$ 
 }
	\KwOut{The optimal objective value $\optvalue$} \BlankLine
	$cost(s)\leftarrow w(s)$, for all disks $s\in S$\;
	\For{$i\leftarrow 1$ \KwTo $n$}
	{
        $\delta_i\leftarrow \min_{s\in S_{p_i}} cost(s)$\tcp*[r]{\findmin\ Operation}         
        \label{ln:findmin}
        $cost(s)\leftarrow w(s) + \delta_i$ for all disks $s\in \overline{S_{p_i}}$\tcp*[r]{\reset\ Operation} 
        \label{ln:reset}
	}
    \Return $\delta_n$\; 
\end{algorithm}
\vspace{-0.2in}

The algorithm is essentially a dynamic program. 
We prove in Section~\ref{sec:correct} that the value $\delta_n$ returned by the algorithm is equal to $\optvalue$, the optimal objective value. To find an optimal subset, one could slightly modify the algorithm following the standard dynamic programming backtracking technique. More specifically, if $\delta_n$ is equal to $cost(s)$ for some disk $s\in S_{p_n}$, then $s$ should be reported as a disk in the optimal subset. Suppose that $cost(s)$ is equal to $w(s)+\delta_i$ for some point $p_i$. Then $\delta_i$ is equal to $cost(s')$ for some disk $s'\in S_{p_i}$ and $s'$ should be reported as a disk in the optimal subset. We continue this backtracking process until a disk whose cost is equal to its own weight is reported (in which case all points of $P$ are covered by the reported disks). 

For reference purposes, we use \findmin\ to refer to the operation in Line~\ref{ln:findmin} and use \reset\ to refer to the operation in Line~\ref{ln:reset} of Algorithm~\ref{algo:10}. The efficiency of the algorithm hinges on how to implement these two {\em key operations}, which will be discussed in Section~\ref{sec:implementation}.

\subsection{Algorithm correctness}
\label{sec:correct}
We now prove that Algorithm~\ref{algo:20} is correct, i.e., prove $\delta_n=\optvalue$. To this end, we show that our algorithm is consistent with the algorithm of Pederson and Wang~\cite{ref:PedersenAl22} for the same problem, or alternatively, our algorithm provides a different implementation of their algorithm. Their algorithm first reduces the problem to a 1D interval coverage problem and then solves the interval coverage problem by a dynamic programming algorithm. In the following, we first review their problem reduction and then explain their dynamic programming algorithm. Finally, we show that our algorithm is essentially an implementation of their dynamic programming algorithm. 

\subsubsection{Reducing the problem to an interval coverage problem}

For convenience, let $p_0$ (resp., $p_{n+1}$) be a point to the left (resp., right) all the points of $P$ and is not contained in any disk of $S$.

Consider a disk $s\in S$. We say that a subsequence $P[i,j]$ of $P$ with $1\leq i\leq j\leq n$ is a {\em maximal subsequence covered} by $s$ if all points of $P[i,j]$ are covered by $s$ but neither $p_{i-1}$ nor $p_{j+1}$ is (it is well defined due to $p_0$ and $p_{n+1}$).
Define $F(s)$ as the set of all maximal subsequences covered by $s$. It is easy to see that the subsequences of $F(s)$ are pairwise disjoint.


For each point $p_i$ of $P$, we vertically project it on the $x$-axis $\ell$; let $p_i^*$ denote the projection point. Let $P^*$ denote the set of all such projection points. 
Due to our general position assumption that no two points of $P$ have the same $x$-coordinate, all points of $P^*$ are distinct. 
For any $1\leq i\leq j\leq n$, we use $P^*[i,j]$ to denote the subsequence $p^*_i,p^*_{i+1},\ldots,p^*_j$.

Next, we define a set $S^*$ of line segments on $\ell$ as follows. For each disk $s\in S$ and each
subsequence $P[i,j]\in F(s)$, we create a segment for $S^*$, denoted
by $s^*[i,j]$, with the left endpoint at $p_i^*$ and the right endpoint at $p^*_j$. As such,
$s^*[i,j]$ covers all points of $P^*[i,j]$ and does not cover any point of $P^*\setminus P^*[i,j]$.
We let the weight of $s^*[i,j]$ be equal to $w(s)$. 
We say that $s^*[i,j]$ is {\em defined} by the disk $s$.

Consider the following {\em interval coverage} problem (i.e., each segment of $S^*$ can also be considered an interval of $\ell$): Find a subset of segments of $S^*$ of minimum total weight such that the union of the segments in the subset covers all points of $P^*$. Pederson and Wang~\cite{ref:PedersenAl22} proved that an optimal solution to this interval coverage problem corresponds to an optimal solution to the original disk coverage problem on $P$ and $S$. More specifically, suppose that $\opt^*$ is an optimal subset of the interval coverage problem. Then, we can obtain an optimal subset $\opt$ for the disk coverage problem as follows: For each segment $s^*[i,j]\in \opt^*$, we add the disk that defines $s^*[i,j]$ to $\opt$. It is proved in \cite{ref:PedersenAl22} that $\opt$ thus obtained is an optimal subset of the disk coverage problem. Note that since a disk of $S$ may define multiple
segments of $S^*$, a potential issue with $\opt$ is that a disk may be included in $\opt$ multiple times (i.e., if multiple segments defined by the disk are in $\opt^*$); but this is proved impossible~\cite{ref:PedersenAl22}.

\subsubsection{Solving the interval coverage problem}
We now explain the dynamic programming algorithm in \cite{ref:PedersenAl22} for the interval coverage problem. 

Let $p_0^*$ be the vertical projection of $p_0$ on $\ell$. Note that $p_0^*,p_1^*,\ldots,p_{n}^*$ are sorted on $\ell$ from left to right. For each segment $s^*\in S^*$, let $w(s^*)$ denote the weight of $s^*$.  

For each segment $s^*\in S^*$, define $f_{s^*}$ as the index of the rightmost point of $P^*\cup\{p_0^*\}$ strictly to the left of the left endpoint of $s^*$. 
Note that $f_{s^*}$ is well defined due to $p_0^*$.


For each $i\in [1,n]$, define $\delta_i^*$ as the minimum total weight of a subset of segments of $S^*$ whose union covers all points of $P^*[1,i]$. The goal of the interval coverage problem is to compute $\delta_n^*$, which is equal to $\optvalue$~\cite{ref:PedersenAl22}. For convenience, we let $\delta^*_0=0$.
For each segment $s^*\in S^*$, define $cost(s^*)=w(s^*)+\delta^*_{f_{s^*}}$. 
One can verify that $\delta_i^*=\min_{s^*\in S^*_{p_i^*}}cost(s^*)$, where $S^*_{p_i^*}\subseteq S^*$ is the subset of segments that cover $p_i^*$. This is the recursive relation of the dynamic programming algorithm.

Assuming that the indices $f_{s^*}$ for all disks $s^*\in S^*$ have been computed, the algorithm works as follows. We sweep a point $q$ on $\ell$ from left to right. Initially, $q$ is at $p_0^*$. During the sweeping, we maintain the subset $S^*_q\subseteq S^*$ of segments that cover $q$. The algorithm maintains the invariant that the cost of each segment of $S^*_q$ is already known and the values $\delta_i^*$ for all points $p_i^*\in P^*$ to the left of $q$ have been computed. An event happens when $q$ encounters an endpoint of a segment of $S^*$ or a point of $P^*$. 
If $q$ encounters a point $p_i^*\in P^*$, then we find the segment of $S^*_q$ with the minimum cost and assign the cost value to $\delta_i^*$. If $q$ encounters the left endpoint of a segment $s^*$, we set $cost(s^*)=w(s^*)+\delta^*_{f_{s^*}}$ and then insert $s^*$ into $S^*_q$. If $q$ encounters the right endpoint of a segment, we remove the segment from $S^*_q$. The algorithm finishes once $q$ meets $p_n^*$, at which event $\delta_n^*$ is computed. 


\paragraph{\bf Remark.} 
It was shown in~\cite{ref:PedersenAl22} that the above dynamic programming algorithm for the interval coverage problem can be implemented in $O((|P^*|+|S^*|)\log (|P^*|+|S^*|))$. While $|P^*|=n$, $|S^*|$ may be relatively large. A straightforward upper bound for $|S^*|$ is $O(nm)$. Pederson and Wang~\cite{ref:PedersenAl22} proved another bound $|S^*|=O(n+m+\kappa)$, where $\kappa$ is the number of pairs of disks that intersect. This leads to their algorithm of $O(nm\log (m+n)+ n\log n)$ time or $O((m+n)\log(m+n) + \kappa\log m)$ time for the original disk coverage problem on $P$ and $S$. 

\subsubsection{Correctness of our algorithm}

Next, we argue that $\delta_n=\delta_n^*$, which will establish the correctness of Algorithm~\ref{algo:10} since $\delta^*_n=\optvalue$. In fact, we will show that $\delta_i=\delta_i^*$ for all $1\leq i\leq n$. We prove it by induction. 

As the base case, we first argue $\delta_1=\delta_1^*$ . To see this, by definition, $\delta_1=\min_{s\in S_{p_1}} w(s)$ because $cost(s)=w(s)$ initially for all disks $s\in S$. For $\delta_1^*$, notice that $f_{s^*}=0$ for every segment $s^*\in S^*_{p^*_1}$. Since $\delta^*_0=0$, we have $\delta_1^*= \min_{s^*\in S^*_{p_1^*}} w(s^*)$. By definition, a segment $s^*\in S^*$ covers $p_1^*$ only if the disk of $S$ defining $s^*$ covers $p_1$, and the segment has the same weight as the disk. Therefore, $s^*$ is in $S^*_{p_1^*}$ only if the disk of $S$ defining $s^*$ is in $S_{p_1}$. On the other hand, if a disk $s$ covers $p_1$, then $s$ must define exactly one segment in $S^*$ covering $p_1^*$. Hence, for each disk $s\in S_{p_1}$, it defines exactly one segment in $S^*_{p_1^*}$ with the same weight. This implies that $\delta_1^*$ is equal to the minimum weight of all disks of $S$ covering $p_1$, and therefore, $\delta_1^*=\delta_1$ must hold. 

Consider any $i$ with $2\leq i\leq n$. Assuming that $\delta_{j}=\delta^*_{j}$ for all $1\leq j<i$, we now prove $\delta_i=\delta^*_i$. Recall that $\delta_i= \min_{s\in S_{p_i}} cost(s)$ and $\delta^*_i= \min_{s^*\in S^*_{p_i^*}} cost(s^*)$. As argued above, each disk $s\in S_{p_i}$ defines a segment in $S^*_{p_i^*}$ with the same weight and each segment $s^*\in S^*_{p_i^*}$ is defined by a disk in $S_{p_i}$ with the same weight. Let $s^*$ be the segment of $S^*_{p_i^*}$ defined by a disk $s\in S_{p_i}$. To prove $\delta_i=\delta_i^*$, it suffices to show that $cost(s)$ of $s$ is equal to $cost(s^*)$ of $s^*$. To see this, first note that $w(s)=w(s^*)$. By definition, $cost(s^*)=w(s^*)+\delta^*_{f_{s^*}}$. For notational convenience, let $j=f_{s^*}$. By definition, all points of $P^*[j+1,i]$ are covered by the segment $s^*$ but the point $p_j^*$ is not. Therefore, all points of $P[j+1,i]$ are covered by the disk $s$ but $p_j$ is not. As such, during the \reset\ operation of the $j$-th iteration of the for loop in Algorithm~\ref{algo:10}, $cost(s)$ will be set to $w(s)+\delta_j$; furthermore, $cost(s)$ will not be reset again during the $i'$-th iteration for all $j+1\leq i'\leq i$. Therefore, we have $cost(s)=w(s)+\delta_j$ at the beginning of the $i$-th iteration of the algorithm. Since $\delta_j=\delta^*_j$ holds by induction hypothesis and $w(s)=w(s^*)$, we obtain $cost(s)=cost(s^*)$. This proves $\delta_i=\delta_i^*$. 

The correctness of Algorithm~\ref{algo:10} is thus established. 

\paragraph{\bf Remark.} The above proof for $\delta_n=\optvalue$ also implies that $\delta_i=\optvalue^i$ for all $1\leq i\leq n-1$, where $\optvalue^i$ is the minimum total weight of a subset of disks whose union covers all points of $P[1,i]$. Indeed, we can  apply the same proof to the points of $P[1,i]$ only. Observe that $\delta_i$ will never change after the $i$-th iteration of Algorithm~\ref{algo:10}.

\section{Algorithm implementation}
\label{sec:implementation}
In this section, we discuss the implementation of Algorithm~\ref{algo:10}. Specifically, we describe how to implement the two key operations \findmin\ and \reset. A straightforward method can implement each operation in $O(m)$ time, resulting in a total $O(mn+n\log n)$ time of the algorithm. Note that this is already a logarithmic factor improvement over the previous work of Pederson and Wang~\cite{ref:PedersenAl22}. In the following, we present a faster approach of $O(n\sqrt{m}\log^2 m +(m+n)\log(m+n))$ time. 

\paragraph{\bf Duality.}
Recall that the points of $P$ are sorted from left to right as $p_1,p_2,\ldots,p_n$.
In fact, we consider the problem in the ``dual'' setting. For each point $p_i\in P$, let $d_i$ denote the unit disk centered at $p_i$, and we call $d_i$ the {\em dual disk} of $p_i$. For each disk $s\in S$, let $q_s$ denote the center of $s$, and we call $q_s$ the {\em dual point} of $s$. We define the weight of $q_s$ to be equal to $w(s)$. 
We use $D$ to denote the set of all dual disks and $Q$ the set of all dual points. 
For each dual point $q\in Q$, let $w(q)$ denote its weight. 
Because all disks of $S$ are unit disks, we have the following observation. 
\begin{observation}\label{obser:dual}
    A disk $s\in S$ covers a point $p_i\in P$ if and only if the dual point $q_s$ is covered by the dual disk $d_i$.
\end{observation}

For any disk $d_i\in D$, let $Q_{d_i}$ denote the subset of dual points of $Q$ that are covered by $d_i$, i.e., $Q_{d_i}=Q\cap d_i$. Define $\overline{Q_{d_i}}=Q\setminus Q_{d_i}$. In light of Observation~\ref{obser:dual}, Algorithm~\ref{algo:10} is equivalent to the following Algorithm~\ref{algo:20}. 

\begin{algorithm}[h]
	\caption{An algorithm ``dual'' to Algorthm~\ref{algo:10}}
	\label{algo:20}
	\SetAlgoNoLine
    \BlankLine
	$cost(q)\leftarrow w(q)$, for all dual points $q\in Q$\;
	\For{$i\leftarrow 1$ \KwTo $n$}
	{
        $\delta_i\leftarrow \min_{q\in Q_{d_i}} cost(q)$\tcp*[r]{\findmin\ Operation}         
        \label{ln:findmin}
        $cost(q)\leftarrow w(q) + \delta_i$ for all dual points $q\in \overline{Q_{d_i}}$\tcp*[r]{\reset\ Operation} 
        \label{ln:reset}
	}
    \Return $\delta_n$\; 
\end{algorithm}

In the following, we will present an implementation for Algorithm~\ref{algo:20}, and in particular, for the two operations \findmin\ and \reset. 

For each disk $d_i\in D$, since its center is above the $x$-axis $\ell$ and all points of $Q$ are below $\ell$, only the portion of $d_i$ below $\ell$ matters to Algorithm~\ref{algo:20}.
We call the boundary portion of $d_i$ below $\ell$ the {\em lower arc} of $d_i$. 
Let $H$ denote the set of the lower arcs of all disks of $D$. 

\paragraph{\bf Cuttings.}
Our algorithm will need to construct a cutting on the arcs of $H$~\cite{ref:ChazelleCu93}. We explain this concept first. Note that $|H|=n$. 
For a parameter $r$ with $1 \leq r \leq n$, a {\em $(1/r)$-cutting} $\Xi$ of size $O(r^2)$ for $H$ is a collection of $O(r^2)$ constant-complexity cells whose union covers the entire plane and whose interiors are pairwise disjoint such that
the interior of each cell $\sigma\in \Xi$ is crossed by at most $n / r$ arcs of $H$, i.e.,
$|H_{\sigma}| \leq n/r$, where $H_{\sigma}$ is the subset of arcs of $H$ that cross the interior of $\sigma$ ($H_{\sigma}$ is often called the {\em conflict list} in the literature). Let $D_{\sigma}$ be the subset of disks of $D$ whose lower arcs are in $H_{\sigma}$. Hence, we also have $|D_{\sigma}|\leq n/r$.

We actually need to construct a \emph{hierarchical cutting} for $H$~\cite{ref:ChazelleCu93}. We say that a cutting $\Xi'$ \emph{$c$-refines} another cutting $\Xi$ if each cell of $\Xi'$ is completely inside a single cell of $\Xi$ and each cell of $\Xi$ contains at most $c$ cells of $\Xi'$. Let $\Xi_0$ denote the cutting with a single cell that is the entire plane. We define cuttings $\{\Xi_0, \Xi_1, ..., \Xi_k\}$, in which each $\Xi_i$, $1 \leq i \leq k$, is a $(1/\rho^i)$-cutting of size $O(\rho^{2i})$ that $c$-refines $\Xi_{i - 1}$, for two constants $\rho$ and $c$. By setting $k = \lceil \log_\rho r \rceil$, the last cutting $\Xi_k$ is a $(1/r)$-cutting. The sequence $\{\Xi_0, \Xi_1, ..., \Xi_k\}$ is called a {\em hierarchical $(1/r)$-cutting} for $H$. If a cell $\sigma'$ of $\Xi_{i - 1}$, $1 \leq i \leq k$, contains cell $\sigma$ of $\Xi_i$, we say that $\sigma'$ is the \emph{parent} of $\sigma$ and $\sigma$ is a \emph{child} of $\sigma'$.  We can also define {\em ancestors} and {\em descendants} correspondingly.
As such, the hierarchical $(1/r)$-cutting can be viewed as a tree structure with the single cell of $\Xi_0$ as the root. We often use $\Xi$ to denote the set of all cells in all cuttings $\Xi_i$, $0\leq i\leq k$.
The total number of cells of $\Xi$ is $O(r^2)$~\cite{ref:ChazelleCu93}. 

\begin{figure}[h]
\begin{minipage}[t]{\textwidth}
\begin{center}
\includegraphics[height=1.0in]{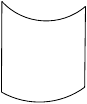}
\caption{\footnotesize Illustrating a pseudo-trapezoid.}
\label{fig:pseudotrap}
\end{center}
\end{minipage}
\vspace{-0.15in}
\end{figure}

A hierarchical $(1/r)$-cutting of $H$ can be computed in $O(nr)$ time, e.g., by the algorithm in \cite{ref:WangUn23}, which adapts Chazelle's algorithm~\cite{ref:ChazelleCu93} for hyperplanes. The algorithm also produces the conflict lists $H_{\sigma}$ (and thus $D_{\sigma}$) for all cells $\sigma\in \Xi$, implying that the total size of these conflict lists is bounded by $O(nr)$. In particular, each cell of the cutting produced by the algorithm of \cite{ref:WangUn23} is a (possibly unbounded) {\em pseudo-trapezoid} that typically has two vertical line segments as left and right sides, a sub-arc of an arc of $H$ as a top side (resp., bottom side) (see Fig.~\ref{fig:pseudotrap}).

In what follows, we first discuss a preprocessing step in Section~\ref{sec:preprocess}. The algorithms for handling the two key operations are described in the subsequent two subsections, respectively. Section~\ref{sec:summary} finally summarizes everything. 

\subsection{Preprocessing}
\label{sec:preprocess}
In order to handle the two key operations, we first perform some preprocessing work before we run Algorithm~\ref{algo:20}. As discussed above, we first sort all points of $P$ from left to right. In the following, we describe a data structure which will be used to support the two key operations. 

We start by computing a hierarchical $(1/r)$-cutting $\{\Xi_0, \Xi_1, ..., \Xi_k\}$ for $H$ in $O(nr)$ time~\cite{ref:ChazelleCu93,ref:WangUn23}, for a parameter $1\leq r\leq n$ to be determined later. We follow the above notation, e.g., $\sigma$, $H_\sigma$, $D_{\sigma}$, $\Xi$. As discussed above, the cutting algorithm also produces the conflict lists $H_{\sigma}$ (and thus $D_{\sigma}$) for all cells $\sigma\in \Xi$. Using the conflict lists, we compute a list $L(d_i)$ for each disk $d_i\in D$, where $L(d_i)$ comprises all cells $\sigma\in \Xi$ such that $d_i\in D_{\sigma}$. Computing $L(d_i)$ for all disks $d_i\in D$ can be done in $O(\sum_{\sigma\in \Xi}|H_{\sigma}|)$ time by simply traversing the conflict lists $H_{\sigma}$ of all cells $\sigma\in \Xi$, which takes $O(nr)$ time as $\sum_{\sigma\in \Xi}|H_{\sigma}|=O(nr)$. Note that this also implies $\sum_{d_i\in D}|L(d_i)|=O(nr)$.

For any region $R$ in the plane, let $Q(R)$ denote the subset of points of $Q$ that are inside $R$, i.e., $Q(R)=Q\cap R$. 

For simplicity, we assume that no point of $Q$ is on the boundary of any cell of $\Xi$. This implies that 
each point of $Q$ is in the interior of a single cell of $\Xi_i$, for all $0\leq i\leq k$. 
We compute the subset $Q(\sigma)$ of all cells $\sigma$ in the last cutting $\Xi_k$. This can be done by a point location procedure as follows. For each point $q\in Q$, starting from the only cell of $\Xi_0$, we find the cell $\sigma_i$ of $\Xi_i$ containing $q$, for each $0\leq i\leq k$. More precisely, suppose that $\sigma_i$ is known. To find $\sigma_{i+1}$, we simply check all $O(1)$ children of $\sigma_i$ and find the one that contains $q$, which takes $O(1)$ time. As such, processing all points of $Q$ takes $O(m\log r)$ time as $k=O(\log r)$, after which $Q(\sigma)$ for all cells $\sigma\in \Xi_k$ are computed. We explicitly store $Q(\sigma)$ for all cells $\sigma\in \Xi_k$. Note that the subsets $Q(\sigma)$ for all cells $\sigma\in \Xi_k$ form a partition of $Q$. Therefore, we have the following observation. 

\begin{observation}\label{obser:partition}
$\sum_{\sigma\in \Xi_k}|Q(\sigma)|=m$.
\end{observation}



Note that a cell $\sigma\in \Xi$ is the ancestor of another cell $\sigma'\in \Xi$ (alternatively, $\sigma'$ is a descendant of $\sigma$) if and only if $\sigma$ fully contains $\sigma'$. 
For convenience, we consider $\sigma$ an ancestor of itself but not a descendant of itself.  
Let $A(\sigma)$ denote the set of all ancestors of $\sigma$ and $B(\sigma)$ the set of all descendants of $\sigma$.
Hence, $\sigma$ is in $A(\sigma)$ but not in $B(\sigma)$. Let $C(\sigma)$ denote the set of all children of $\sigma$
Clealy, $|A(\sigma)|=O(\log r)$ and $|C(\sigma)|=O(1)$. 

\paragraph{\bf Variables and algorithm invariants.}
For each point $q\in Q$, we associate with it a variable $\lambda(q)$. 
For each cell $\sigma\in \Xi$, we associate with it two variables: $minCost(\sigma)$ and $\lambda(\sigma)$.
If $|Q(\sigma)|=\emptyset$, then $minCost(\sigma)=\infty$ and $\lambda(\sigma)=0$ always hold during the algorithm. 
Our algorithm for handling the two key operations will maintain the following two invariants. 
\begin{invariant}
For any point $q\in Q$, $cost(q)=w(q)+\lambda(q)+\sum_{\sigma'\in A(\sigma)}\lambda(\sigma')$, where $\sigma$ is the cell of $\Xi_k$ that contains $q$. 
\end{invariant}
\begin{invariant}
For each cell $\sigma\in \Xi$ with $Q(\sigma)\neq \emptyset$, if $\sigma$ is a cell of $\Xi_k$, then $minCost(\sigma)=\min_{q\in Q(\sigma)}(w(q)+\lambda(q))$; otherwise, $minCost(\sigma)=\min_{\sigma'\in C(\sigma)}(minCost(\sigma')+\lambda(\sigma'))$. 
\end{invariant}

The above algorithm invariants further imply the following observation. 
\begin{observation}\label{obser:invariant}
For each cell $\sigma\in \Xi$ with $Q(\sigma)\neq \emptyset$, $\min_{q\in Q(\sigma)}cost(q)=minCost(\sigma)+\sum_{\sigma'\in A(\sigma)}\lambda(\sigma')$.    
\end{observation}
\begin{proof}
We prove the observation by induction. For the base case, consider a cell $\sigma$ of the last cutting $\Xi_k$. By the two algorithm invariants, we have 
\begin{equation*}
\begin{aligned}
\min_{q\in Q(\sigma)}cost(q)&=\min_{q\in Q(\sigma)}\Bigl(w(q)+\lambda(q)+\sum_{\sigma'\in A(\sigma)}\lambda(\sigma')\Bigr) && \text{by Algorithm Invariant 1}\\
&=\min_{q\in Q(\sigma)}\Bigl(w(q)+\lambda(q)\Bigr)+\sum_{\sigma'\in A(\sigma)}\lambda(\sigma')\\
&=minCost(\sigma)+\sum_{\sigma'\in A(\sigma)}\lambda(\sigma'). && \text{by Algorithm Invariant 2}         
\end{aligned}
\end{equation*}
This proves the observation for $\sigma$. 

Now consider a cell $\sigma\in \Xi\setminus\Xi_k$. We assume that the observation holds for all children $\sigma'$ of $\sigma$, i.e., $\min_{q\in Q(\sigma')}cost(q)=minCost(\sigma')+\sum_{\sigma''\in A(\sigma')}\lambda(\sigma'')$. Then, we have 
\begin{align*}
\min_{q\in Q(\sigma)}cost(q) &=\min_{\sigma'\in C(\sigma)}\min_{q\in Q(\sigma')}cost(q)\\
&=\min_{\sigma'\in C(\sigma)}\Bigl(minCost(\sigma')+\sum_{\sigma''\in A(\sigma')}\lambda(\sigma'')\Bigr) && \text{by induction hypothesis} \\
&=\min_{\sigma'\in C(\sigma)}\Bigl(minCost(\sigma')+\lambda(\sigma')\Bigr)+\sum_{\sigma''\in A(\sigma)}\lambda(\sigma'')\\
&=minCost(\sigma)+\sum_{\sigma''\in A(\sigma)}\lambda(\sigma'').  && \text{by Algorithm Invariant 2}   
\end{align*}
This proves the observation. 
\qed
\end{proof}

For each cell $\sigma\in \Xi$, we also maintain $\calL(\sigma)$, a list comprising all descendant cells $\sigma'$ of $\sigma$ whose values $\lambda(\sigma')$ are not zero and all points $q\in Q(\sigma)$ whose values $\lambda(q)$ are not zero. As $\calL(\sigma)$ has both cells of $B(\sigma)$ and points of $Q(\sigma)$, for convenience, we use an ``element'' to refer to either a cell or a point of $\calL(\sigma)$.
As such, for any element $e\in B(\sigma)\cup Q(\sigma)$ with $e\not\in \calL(\sigma)$, $\lambda(e)=0$ must hold. 
As will be seen later, whenever the algorithm sets $\lambda(\sigma)$ to a nonzero value for a cell $\sigma\in \Xi$,  $\sigma$ will be added to $\calL(\sigma')$ for every ancestor $\sigma'$ of $\sigma$ with $\sigma'\neq \sigma$. Similarly, whenever the algorithm sets $\lambda(q)$ to a nonzero value for a point $q$, then $q$ will be added to $\calL(\sigma')$ for every cell $\sigma'\in A(\sigma)$, where $\sigma$ is the cell of $\Xi_k$ containing $q$. 


\paragraph{\bf Initialization.}
The above describes the data structure. We now initialize the data structure, and in particular, initialize the variables $\calL(\cdot)$, $\lambda(\cdot)$, $minCost(\cdot)$ so that the algorithm invariants hold. 

First of all, for each cell $\sigma\in \Xi$, we set $\calL(\sigma)=\emptyset$ and $\lambda(\sigma)=0$. For each point $q\in Q$, we set $\lambda(q)=0$. Since $cost(q)=w(q)$ initially, it is not difficult to see that Algorithm Invariant 1 holds. 

We next set $minCost(\sigma)$ for all cells of $\sigma\in \Xi$ in a bottom-up manner following the tree structure of $\Xi$. Specifically, for each cell $\sigma$ in the last cutting $\Xi_k$, we set $minCost(\sigma)=\min_{q\in Q(\sigma)}w(q)$ by simply checking every point of $Q(\sigma)$. If $Q(\sigma)=\emptyset$, we set $minCost(\sigma)=\infty$. 
This establishes the second algorithm invariant for all cells $\sigma\in \Xi_k$. 
Then, we set $minCost(\sigma)$ for all cells of $\sigma\in \Xi_{k-1}$ with $minCost(\sigma)=\min_{\sigma'\in C(\sigma)}(minCost(\sigma')+\lambda(\sigma'))$, after which the second algorithm invariant holds for all cells $\sigma\in \Xi_{k-1}$. We continue this process to set $minCost(\sigma)$ for cells in $\Xi_{k-2},\Xi_{k-3},\ldots,\Xi_0$. 
After that, the second algorithm invariant is established for all cells $\sigma\in \Xi$. 

In addition, for each cell $\sigma$ in the last cutting $\Xi_k$, in order to efficiently update $minCost(\sigma)$ once $\lambda(q)$ changes for a point $q\in Q(\sigma)$, we construct a min-heap $\calH(\sigma)$ on all points $q$ of $Q(\sigma)$ with the values $w(q)+\lambda(q)$ as ``keys''. Using the heap, if $\lambda(q)$ changes for a point $q\in Q(\sigma)$, $minCost(\sigma)$ can be updated in $O(\log m)$ time as $|Q(\sigma)|\leq m$. 

This finishes our preprocessing step for Algorithm~\ref{algo:20}. The following lemma analyzes the time complexity of the preprocessing. 

\begin{lemma}\label{lem:timepre}
The preprocessing takes $O(n\log n+nr+m\log r)$ time. 
\end{lemma}
\begin{proof}
First of all, sorting $P$ takes $O(n\log n)$ time. Constructing the hierarchical cutting $\Xi$ takes $O(nr)$ time. Computing the lists $L(d_i)$ for all disks $d_i\in D$ also takes $O(nr)$ time. The point location procedure for computing the subsets $Q(\sigma)$ for all cells $\sigma\in \Xi_k$ runs in $O(m\log r)$ time. For the initialization step, setting $\calL(Q)=\emptyset$ and $\lambda(\sigma)=0$ for all cells $\sigma\in \Xi$ takes $O(r^2)$ time as $|\Xi|=O(r^2)$. Since $\Xi_k$ has $O(r^2)$ cells, computing $minCost(\sigma)$ for all cells $\sigma\in \Xi_k$ can be done in $O(r^2+\sum_{\sigma\in\Xi_k}|Q(\sigma)|)$ time, which is $O(r^2+m)$ by Observation~\ref{obser:partition}. Initializing $minCost(\sigma)$ for all other cells $\sigma\in \Xi\setminus \Xi_k$ takes $O(r^2)$ time since each cell has $O(1)$ children (and thus computing $minCost(\sigma)$ for each such cell $\sigma$ takes $O(1)$ time). Finally, constructing a heap $\calH(\sigma)$ for all cells $\sigma\in \Xi_k$ takes $O(\sum_{\sigma\in\Xi_k}|Q(\sigma)|)$ time, which is $O(m)$ by Observation~\ref{obser:partition}. Since $r\leq n$, $r^2\leq nr$. Therefore, the total time of the preprocessing is $O(n\log n+nr+m\log r)$. 
\qed 
\end{proof}


\subsection{The \findmin\ operation}

We now discuss how to perform the \findmin\ operation. 

Consider a disk $d_i$ in \findmin\ operation of the $i$-th iteration of Algorithm~\ref{algo:20}. The goal is to compute $\min_{q\in Q_{d_i}}cost(q)$, i.e., the minimum cost of all points of $Q$ inside the disk $d_i$. 

Recall that $L(d_i)$ is the list of all cells $\sigma\in \Xi$ such that $d_i\in D_{\sigma}$.
Define $L_1(d_i)$ to be the set of all cells of $L(d_i)$ that are from $\Xi_k$ and let $L_2(d_i)=L(d_i)\setminus L_1(d_i)$. Define $L_3(d_i)$ as the set of cells $\sigma\in \Xi$ such that $\sigma$'s parent is in $L_2(d_i)$ and $\sigma$ is completely contained in $d_i$. 
We first have the following observation following the definition of the hierarchical cutting. 

\begin{observation}\label{obser:Qunion}
$Q_{d_i}$ is the union of $\bigcup_{\sigma\in L_1(d_i)}(Q(\sigma)\cap d_i)$ and $\bigcup_{\sigma\in L_3(d_i)}Q(\sigma)$.
\end{observation}
\begin{proof}
Consider a point $q\in \bigcup_{\sigma\in L_1(d_i)}(Q(\sigma)\cap d_i)$. Suppose that $q$ is in $Q(\sigma)\cap d_i$ for some cell $\sigma\in L_1(d_i)$. Then, since $q\in d_i$, it is obvious true that $q\in Q_{d_i}$. 

Consider a point $q\in \bigcup_{\sigma\in L_3(d_i)}Q(\sigma)$. Suppose that $q\in Q(\sigma)$ for some cell $\sigma\in L_3(d_i)$. By the definition of $L_3(d)$, $\sigma$ is fully contained in $d_i$. Therefore, $q\in d_i$ holds. Hence, $q\in Q_{d_i}$. 

On the other hand, consider a point $q\in Q_{d_i}$. By definition, $d_i$ contains $q$. Let $\sigma$ be the cell of $\Xi_k$ containing $q$. Since both $d_i$ and $\sigma$ contain $q$, $d_i\cap \sigma\neq \emptyset$. Therefore, either $\sigma\subseteq d_i$ or the boundary of $d_i$ crosses $\sigma$. In the latter case, we have $\sigma\in L_1(d_i)$ and thus $q\in \bigcup_{\sigma\in L_1(d_i)}(Q(\sigma)\cap d_i)$. In the former case, $\sigma$ must have two ancestors $\sigma_1$ and $\sigma_2$ such that (1) $\sigma_1$ is the parent of $\sigma_2$; (2) $\sigma_2$ is fully contained in $d_i$; (3) the boundary of $d_i$ crosses $\sigma_1$. This is true because $\sigma$ is fully contained in $d_i$ while the boundary of $d_i$ crosses the only cell of $\Xi_0$, which is the entire plane and is an ancestor of $\sigma$. As such, $\sigma_1$ must be in $L_2(d_i)$ and $\sigma_2$ must be in $L_3(d_i)$. Therefore, $q$ must be in $\bigcup_{\sigma\in L_3(d_i)}Q(\sigma)$. 

This proves the observation. 
\qed
\end{proof}

With Observation~\ref{obser:Qunion}, we now describe our algorithm for \findmin. 
Let $\alpha$ be a variable, which is initialized to $\infty$. At the end of the algorithm, we will have $\alpha=\min_{q\in Q_{d_i}}cost(q)$. For each cell $\sigma$ in the list $L(d_i)$, if it is from $\sigma\in \Xi_k$, i.e., $\sigma\in L_1(d_i)$, then we process $\sigma$ as follows. For each point $q\in Q(\sigma)$, by Algorithm Invariant 1, we have $cost(q)=w(q)+\lambda(q)+\sum_{\sigma'\in A(\sigma)}\lambda(\sigma')$. If $q\in d_i$, we compute $cost(q)$ by visiting all cells of $A(\sigma)$, which takes $O(\log r)$ time, and then we update $\alpha=\min\{\alpha,cost(q)\}$. 

If $\sigma\in L_2(d_i)$, then we process it as follows. For each child $\sigma'$ of $\sigma$ that is fully contained in $d_i$ (i.e., $\sigma\in L_3(d_i)$), we compute $\alpha_{\sigma'}=minCost(\sigma')+\sum_{\sigma''\in A(\sigma')}\lambda(\sigma'')$ by visiting all cells of $A(\sigma')$, which takes $O(\log r)$ time. By Observation~\ref{obser:invariant}, we have $\alpha_{\sigma'}=\min_{q\in Q(\sigma)}cost(q)$. Then we update $\alpha=\min\{\alpha,\alpha_{\sigma'}\}$. After processing every cell $\sigma\in L(d_i)$ as above, we return $\alpha$, which is equal to $\min_{q\in Q_{d_i}}cost(q)$ according to our algorithm invariants as well as Observation~\ref{obser:Qunion}. This finishes the \findmin\ operation. The following lemma analyzes the runtime of the operation. 

\begin{lemma}\label{lem:timefindmin}
The total time of the \findmin\ operations in the entire Algorithm~\ref{algo:20} is bounded by $O((nr+mn/r)\log r)$. 
\end{lemma}
\begin{proof}
Recall that in each operation we processes cells of $L_1(d_i)$ and cells of $L_2(d_i)$ in different ways. The total time of the operation is the sum of the time for processing $L_1(d_i)$ and the time for processing $L_2(d_i)$. 

For the time for processing $L_1(d_i)$, for each cell $\sigma\in L_1(d_i)$, for each point $q\in Q(\sigma)\cap d_i$, we spend $O(\log r)$ time computing $cost(q)$. Hence, the time for processing cells of $L_1(d_i)$ for each $d_i$ is bounded by $O(\sum_{\sigma\in L_1(d_i)}|Q(\sigma)|\log r)$. 
The total time of processing $L_1(d_i)$ in the entire algorithm is on the order of $\sum_{i=1}^n \sum_{\sigma\in L_1(d_i)}|Q(\sigma)|\cdot \log r=\sum_{\sigma\in \Xi_k}(|D_{\sigma}|\cdot |Q(\sigma)|)\cdot \log r$. 
Recall that $|D_{\sigma}|\leq n/r$ for each cell $\sigma\in \Xi_k$. 
Hence, $\sum_{\sigma\in \Xi_k}(|D_{\sigma}|\cdot |Q(\sigma)|)\leq n/r\cdot \sum_{\sigma\in \Xi_k}|Q(\sigma)|$, which is $O(mn/r)$ by Observation~\ref{obser:partition}. 
Therefore, the total time for processing cells of $L_1(d_i)$ in the entire Algorithm~\ref{algo:20} is $O(mn/r\cdot \log r)$. 

For the time for processing $L_2(d_i)$, for each cell $\sigma\in L_2(d_i)$, for each child $\sigma'$ of $\sigma$, it takes $O(\log r)$ time to compute $\alpha_{\sigma'}$. Since $\sigma$ has $O(1)$ cells, the total time for processing all cells of $L_2(d_i)$ is $O(|L_2(d_i)|\cdot \log r)$. The total time of processing $L_2(d_i)$ in the entire algorithm is on the order of $\sum_{i=1}^n|L_2(d_i)|\cdot \log r$. 
Note that $\sum_{i=1}^n|L_2(d_i)|\leq \sum_{i=1}^n|L(d_i)|=O(nr)$. Therefore, the total time for processing cells of $L_2(d_i)$ in the entire Algorithm~\ref{algo:20} is $O(nr\log r)$. 

Summing up the time for processing $L_1(d_i)$ and $L_2(d_i)$ leads to the lemma. \qed
\end{proof}

\subsection{The \reset\ operation}

We now discuss the \reset\ operation. Consider the \reset\ operation in the $i$-th iteration of Algorithm~\ref{algo:20}. The goal is to reset $cost(q)=w(q)+\delta_i$ for all points $q\in Q$ that are outside the disk $d_i$. To this end, we will update our data structure, and more specifically, update the $\lambda(\cdot)$ and $minCost(\cdot)$ values for certain cells of $\Xi$ and points of $Q$ so that the algorithm invariants still hold. 

Define $L_4(d_i)$ as the set of cells $\sigma\in \Xi$ such that $\sigma$'s parent is in $L_2(d_i)$ and $\sigma$ is completely outside $d_i$. Let $\overline{d_i}$ denote the region of the plane outside the disk $d_i$. We have the following observation, which is analogous to Observation~\ref{obser:Qunion}.

\begin{observation}\label{obser:resetunion}
$\overline{Q_{d_i}}$ is the union of $\bigcup_{\sigma\in L_1(d_i)}(Q(\sigma)\cap \overline{d_i})$ and $\bigcup_{\sigma\in L_4(d_i)}Q(\sigma)$.
\end{observation}
\begin{proof}
The proof is the same as that of Observation~\ref{obser:Qunion} except that we use $\overline{d_i}$ to replace $d_i$ and use $L_4(d_i)$ to replace $L_3(d_i)$. We omit the details. \qed
\end{proof}

Our algorithm for \reset\ works as follows. Consider a cell $\sigma\in L(d_i)$. As for the \findmin\ operation, depending on whether $\sigma$ is from $L_1(d_i)$ or $L_2(d_i)$, we process it in different ways. 

If $\sigma$ is from $L_1(d_i)$, we process $\sigma$ as follows. For each point $q\in Q(\sigma)$, if $q\in \overline{d_j}$, then we are supposed to reset $cost(q)$ to $w(q)+\delta_i$. To achieve the effect and also maintain the algorithm invariants, we do the following. First, we set $\lambda(q)=\delta_i-\sum_{\sigma'\in A(\sigma)}\lambda(\sigma')$, which can be done in $O(\log r)$ time by visiting the ancestors of $\sigma$. As such, we have $w(q)+\lambda(q)+\sum_{\sigma'\in A(\sigma)}\lambda(\sigma')=w(q)+\delta_i$, which establishes the first algorithm invariant for $q$. For the second algorithm invariant, we first update $minCost(\sigma)$ using the heap $\calH(\sigma)$, i.e., by updating the key of $q$ to the new value $w(q)+\lambda(q)$. The heap operation takes $O(\log m)$ time. Next, we update $minCost(\sigma')$ for all ancestors $\sigma'$ of $\sigma$ in a bottom-up manner using the formula $minCost(\sigma')=\min_{\sigma''\in C(\sigma')}(minCost(\sigma'')+\lambda(\sigma''))$. Since each cell has $O(1)$ children, updating all ancestors of $\sigma$ takes $O(\log r)$ time. This establishes the second algorithm invariant. Finally, since $\lambda(q)$ has just been changed, if $\lambda(q)\neq 0$, then we add $q$ to the list $\calL(\sigma')$ for all cells $\sigma'\in A(\sigma)$. Note that for each such $\calL(\sigma')$ it is possible that $q$ was already in the list before; but we do not check this and simply add $q$ to the end of the list (and thus the list may contain multiple copies of $q$). This finishes the processing of $q$, which takes $O(\log r+\log m)$ time. 
Processing all points of $q\in Q(\sigma)$ as above takes $O(|Q(\sigma)|\cdot (\log r+\log m))$ time. 

If $\sigma$ is from $L_2(d_i)$, then we process $\sigma$ as follows. For each child $\sigma'$ of $\sigma$, if $\sigma'$ is completely outside $d_i$, then we process $\sigma'$ as follows. We are supposed to reset $cost(q)$ to $w(q)+\delta_i$ for all points $q\in Q(\sigma')$. In other words, the first algorithm invariant does not hold any more and we need to update our data structure to restore it. Note that the second algorithm invariant still holds. 
To achieve the effect and also maintain the algorithm invariants, we do the following. 
For each element $e$ in the list $\calL(\sigma')$ (recall that $e$ is either a cell of $B(\sigma')$ or a point of $Q(\sigma')$), we process $e$ as follows. First, we remove $e$ from $\calL(\sigma')$. Then we reset $\lambda(e)=0$. If $e$ is a point of $Q(\sigma')$, then let $\sigma_e$ be the cell of $\Xi_k$ that contains $e$; otherwise, $e$ is a cell of $B(\sigma')$ and let $\sigma_e$ be the parent of $e$. Since $\lambda(e)$ is changed, we update $minCost(\sigma'')$ for all cells $\sigma''\in A(\sigma_e)$ in the same way as above in the first case for processing $L_1(d_i)$, which takes $O(\log r+\log m)$ time. This finishes processing $e$, after which the second algorithm invariant still holds. After all elements of $\calL(\sigma')$ are processed as above, $\calL(\sigma')$ becomes $\emptyset$ and we reset $\lambda(\sigma')=\delta_i-\sum_{\sigma''\in A(\sigma')\setminus\{\sigma'\}}\lambda(\sigma'')$. Since $\lambda(\sigma')$ has been changed, we update $minCost(\sigma'')$ for all cells $\sigma''\in A(\sigma)$ in the same way as before (which takes $O(\log r)$ time), after which the second algorithm invariant still holds. In addition, if $\lambda(\sigma')\neq 0$, then we add $\sigma'$ to the list $\calL(\sigma'')$ for all cells $\sigma''\in A(\sigma)$, which again takes $O(\log r)$ time. 
This finishes processing $\sigma'$, which takes $O(|\calL(\sigma')|\cdot (\log r+\log m))$ time. It remains to restore the first algorithm invariant, for which we have the following observation. 
\begin{observation}
After $\sigma'$ is processed, the first algorithm invariant is established for all points $q\in Q(\sigma')$.     
\end{observation}
\begin{proof}
Consider a point $q\in Q(\sigma')$. It suffices to show $w(q)+\delta_i=w(q)+\lambda(q)+\sum_{\sigma''\in A(\sigma_q)}\lambda(\sigma'')$, where $\sigma_q$ is the cell of $\Xi_k$ that contains $q$. After the elements of the list $\calL(\sigma')$ are processed as above, we have $\lambda(q)=0$ for all points $q\in Q(\sigma')$ and $\lambda(\sigma'')=0$ for all descendants $\sigma''$ of $\sigma'$. Therefore, $w(q)+\lambda(q)+\sum_{\sigma''\in A(\sigma_1)}\lambda(\sigma_1)=w(q)+\sum_{\sigma''\in A(\sigma')}\lambda(\sigma'')=w(q)+\lambda(\sigma')+\sum_{\sigma''\in A(\sigma')\setminus\{\sigma'\}}\lambda(\sigma'')$. Recall that $\lambda(\sigma')=\delta_i-\sum_{\sigma''\in A(\sigma')\setminus\{\sigma'\}}\lambda(\sigma'')$. We thus obtain $w(q)+\lambda(q)+\sum_{\sigma''\in A(\sigma_q)}\lambda(\sigma'')=w(q)+\delta_i$. \qed
\end{proof}

This finishes the \reset\ operation. According to Observation~\ref{obser:resetunion}, $cost(q)$ has been reset for all points $q\in Q$ that are outside $d_i$.
The following lemma analyzes the runtime of the operation. 

\begin{lemma}\label{lem:timeresetcost}
The total time of the \reset\ operations in the entire Algorithm~\ref{algo:20} is bounded by $O((nr+mn/r)\cdot \log r\cdot (\log r+\log m))$. 
\end{lemma}
\begin{proof}
Recall that we processes cells of $L_1(d_i)$ and cells of $L_2(d_i)$ in different ways. The total time of the operation is the sum of the time for processing $L_1(d_i)$ and the time for processing $L_2(d_i)$. 

For the time for processing $L_1(d_i)$, for each cell $\sigma\in L_1(d_i)$, recall that processing all points of $Q(\sigma)$ takes $O(|Q(\sigma)|\cdot (\log r+\log m))$ time. Hence, the total time for processing cells of $L_1(d_i)$ is on the order of $\sum_{\sigma\in L_1(d_i)}|Q(\sigma)|\cdot (\log r+ \log m)$. The total time for processing cells of $L_1(d_i)$ in the entire Algorithm~\ref{algo:20} is thus on the order of $\sum_{i=1}^n \sum_{\sigma\in L_1(d_i)}|Q(\sigma)|\cdot (\log r+ \log m)$. 
As analyzed in the proof of Lemma~\ref{lem:timefindmin}, $\sum_{i=1}^n \sum_{\sigma\in L_1(d_i)}|Q(\sigma)|=O(mn/r)$. Therefore, the total time for processing cells of $L_1(d_i)$ in the entire Algorithm~\ref{algo:20} is $O(mn/r\cdot (\log r+\log m))$. 

For the time for processing $L_2(d_i)$, for each cell $\sigma\in L_2(d_i)$, for each child $\sigma'$ of $\sigma$, processing $\sigma'$ takes $O(|\calL(\sigma')|\cdot (\log r+\log m))$ time. 
Next, we give an upper bound for $|\calL(\sigma')|$ for all such cells $\sigma'$ in the entire algorithm. Recall that each element $e$ of $\calL(\sigma')$ is either a point $q\in Q(\sigma')$ or a descendant cell $\sigma''\in B(\sigma')$. Let $\calL_1(\sigma')$ denote the subset of elements of $\calL(\sigma')$ in the former case and $\calL_2(\sigma')$ the subset of elements in the latter case. In the following we provide an upper bound for each subset. 
\begin{enumerate}
    \item 
For $\calL_1(\sigma')$, notice that a point $q$ is in the list only if $\sigma_q$ is crossed by the lower arc of a disk $d_i$, where $\sigma_q$ is the cell of $\Xi_k$ containing $q$. If $q$ is outside $d_i$, then a copy of $q$ will be added to $\calL(\sigma'')$ for all $O(\log r)$ cells $\sigma''$ of $A(\sigma_q)$. As $|D_{\sigma_{q}}|\leq n/r$, the number of elements in $\calL_1(\sigma')$ contributed by the points of $Q(\sigma_{q})$ for all such cells $\sigma'$ in the entire algorithm is bounded by $O(|Q(\sigma_{q})|\cdot n/r\cdot \log r)$. 
In light of Observation~\ref{obser:partition}, the total size of $\calL_1(\sigma')$ of all such cells $\sigma'$ in the entire Algorithm~\ref{algo:20} is $O(mn/r\cdot \log r)$. 

\item 
For $\calL_2(\sigma')$, observe that a cell $\sigma_1$ is in the list only if $\sigma_2$ is crossed by the lower arc of a disk $d_i$, where $\sigma_2$ is the parent of $\sigma_1$. If $\sigma_1$ is completely outside $d_i$, then a copy of $\sigma_1$ is added to $\calL(\sigma'')$ for all $O(\log r)$ cells $\sigma''$ of $A(\sigma_2)$. As such, the number of elements in $\calL_2(\sigma')$ for all such cells $\sigma'$ in the entire algorithm contributed by each cell $\sigma_1\in \Xi$ is bounded by $O(|D_{\sigma_2}|\cdot \log r)$. Since every cell of $\Xi$ has $O(1)$ children and $\sum_{\sigma_2\in \Xi}|D_{\sigma_2}|=O(nr)$, the total size of $\calL_2(\sigma')$ of all such cells $\sigma'$ in the entire Algorithm~\ref{algo:20} is $O(nr\cdot \log r)$. 
\end{enumerate}

Therefore, the total time for processing cells of $L_2(d_i)$ in the entire Algorithm~\ref{algo:20} is $O((nr+mn/r)\cdot \log r\cdot (\log r+\log m))$. 

Summing up the time for processing $L_1(d_i)$ and $L_2(d_i)$ leads to the lemma.\qed
\end{proof}

\subsection{Putting it all together}
\label{sec:summary}

We summarize the time complexity of the overall algorithm. By Lemma~\ref{lem:timepre}, the preprocessing step takes $O(n\log n+nr+m\log r)$ time. By Lemma~\ref{lem:timefindmin}, the total time for performing the \findmin\ operations in the entire algorithm is $O((nr+mn/r)\cdot\log r)$. By Lemma~\ref{lem:timeresetcost}, the total time for performing the \reset\ operations in the entire algorithm is $O((nr+mn/r)\cdot \log r\cdot (\log m+\log r))$. Therefore, the total time of the overall algorithm is $O(n\log n+m\log r+(nr+mn/r)\cdot \log r\cdot (\log m+\log r))$. Recall that $1\leq r\leq n$. Setting $r=\min\{\sqrt{m},n\}$ gives the upper bound $O(n\sqrt{m}\log^2 m + (n+m)\log(n+m))$ for the time complexity of the overall algorithm. 

Note that we have assumed that each point of $P$ is covered by at least one disk of $S$. In this is not the case, then no feasible subset exists (alternatively, one may consider the optimal objective value $\infty$); if we run our algorithm in this case, then one can check that the value $\delta_n$ returned by our algorithm is $\infty$. Hence, our algorithm can automatically determine whether a feasible subset exists.\footnote{It is possible to determine whether a feasible subset exists in $O(n\log n)$ time. For example, one can first compute the upper envelope of the boundary portions of all disks above $\ell$. Then, it suffices to determine whether every point of $P$ is below the upper envelope. Note that the upper envelope is $x$-monotone.}

\begin{theorem}\label{theo:coverage}
Given a set of $n$ points and a set of $m$ weighted unit disks in the plane such that the points and the disk centers are separated by a line, there is an $O(n\sqrt{m}\log^2 m + (n+m)\log(n+m)) $  time algorithm to compute a subset of disks of minimum total weight whose union covers all points. 
\end{theorem}



\bibliographystyle{plainurl}

\begin{thebibliography}{10}

\bibitem{ref:AgarwalNe20}
Pankaj~K. Agarwal and Jiangwei Pan.
\newblock Near-linear algorithms for geometric hitting sets and set covers.
\newblock {\em Discrete and Computational Geometry}, 63:460–--482, 2020.
\newblock \href {https://doi.org/10.1007/s00454-019-00099-6} {\path{doi:10.1007/s00454-019-00099-6}}.

\bibitem{ref:AltMi06}
Helmut Alt, Esther~M. Arkin, Herv\'{e} {Br\"onnimann}, Jeff Erickson, S\'{a}ndor~P. Fekete, Christian Knauer, Jonathan Lenchner, Joseph S.~B. Mitchell, and Kim Whittlesey.
\newblock Minimum-cost coverage of point sets by disks.
\newblock In {\em Proceedings of the 22nd Annual Symposium on Computational Geometry (SoCG)}, pages 449--458, 2006.
\newblock \href {https://doi.org/10.1145/1137856.1137922} {\path{doi:10.1145/1137856.1137922}}.

\bibitem{ref:AmbuhlCo06}
Christoph {Amb\"uhl}, Thomas Erlebach, {Mat\'u\u s} {Mihal\'ak}, and Marc Nunkesser.
\newblock Constant-factor approximation for minimum-weight (connected) dominating sets in unit disk graphs.
\newblock In {\em Proceedings of the 9th International Conference on Approximation Algorithms for Combinatorial Optimization Problems (APPROX), and the 10th International Conference on Randomization and Computation (RANDOM)}, pages 3--14, 2006.
\newblock \href {https://doi.org/10.1007/11830924_3} {\path{doi:10.1007/11830924_3}}.

\bibitem{ref:Ben-OrLo83}
Michael Ben{-}Or.
\newblock Lower bounds for algebraic computation trees (preliminary report).
\newblock In {\em Proceedings of the 15th Annual ACM Symposium on Theory of Computing (STOC)}, pages 80--86, 1983.
\newblock \href {https://doi.org/10.1145/800061.808735} {\path{doi:10.1145/800061.808735}}.

\bibitem{ref:BiloGe05}
Vittorio Bil\`o, Ioannis Caragiannis, Christos Kaklamanis, and Panagiotis Kanellopoulos.
\newblock Geometric clustering to minimize the sum of cluster sizes.
\newblock In {\em Proceedings of the 13th European Symposium on Algorithms (ESA)}, pages 460--471, 2005.
\newblock \href {https://doi.org/10.1007/11561071_42} {\path{doi:10.1007/11561071_42}}.

\bibitem{ref:BiniazFa18}
Ahmad Biniaz, Prosenjit Bose, Paz Carmi, Anil Maheshwari, J.~Ian Munro, and Michiel Smid.
\newblock Faster algorithms for some optimization problems on collinear points.
\newblock In {\em Proceedings of the 34th International Symposium on Computational Geometry (SoCG)}, pages 8:1--8:14, 2018.
\newblock \href {https://doi.org/10.4230/LIPIcs.SoCG.2018.8} {\path{doi:10.4230/LIPIcs.SoCG.2018.8}}.

\bibitem{ref:BusPr18}
Norbert Bus, Nabil~H. Mustafa, and Saurabh Ray.
\newblock Practical and efficient algorithms for the geometric hitting set problem.
\newblock {\em Discrete Applied Mathematics}, 240:25--32, 2018.
\newblock \href {https://doi.org/10.1016/j.dam.2017.12.018} {\path{doi:10.1016/j.dam.2017.12.018}}.

\bibitem{ref:ChanEx14}
Timothy~M. Chan and Elyot Grant.
\newblock Exact algorithms and {APX}-hardness results for geometric packing and covering problems.
\newblock {\em Computational Geometry: Theory and Applications}, 47:112--124, 2014.
\newblock \href {https://doi.org/10.1016/j.comgeo.2012.04.001} {\path{doi:10.1016/j.comgeo.2012.04.001}}.

\bibitem{ref:ChanFa20}
Timothy~M. Chan and Qizheng He.
\newblock Faster approximation algorithms for geometric set cover.
\newblock In {\em Proceedings of 36th International Symposium on Computational Geometry (SoCG)}, pages 27:1--27:14, 2020.
\newblock \href {https://doi.org/10.4230/LIPIcs.SoCG.2020.27} {\path{doi:10.4230/LIPIcs.SoCG.2020.27}}.

\bibitem{ref:ChazelleCu93}
Bernard Chazelle.
\newblock Cutting hyperplanes for divide-and-conquer.
\newblock {\em Discrete \& Computational Geometry}, 9:145--158, 1993.
\newblock \href {https://doi.org/10.1007/BF02189314} {\path{doi:10.1007/BF02189314}}.

\bibitem{ref:ClaudeAn10}
Francisco Claude, Gautam~K. Das, Reza Dorrigiv, Stephane Durocher, Robert Fraser, Alejandro L\'opez-Ortiz, Bradford~G. Nickerson, and Alejandro Salinger.
\newblock An improved line-separable algorithm for discrete unit disk cover.
\newblock {\em Discrete Mathematics, Algorithms and Applications}, 2:77--88, 2010.
\newblock \href {https://doi.org/10.1142/S1793830910000486} {\path{doi:10.1142/S1793830910000486}}.

\bibitem{ref:FederOp88}
Tom\'{a}s Feder and Daniel~H. Greene.
\newblock Optimal algorithms for approximate clustering.
\newblock In {\em Proceedings of the 20th Annual ACM Symposium on Theory of Computing (STOC)}, pages 434--444, 1988.
\newblock \href {https://doi.org/10.1145/62212.62255} {\path{doi:10.1145/62212.62255}}.

\bibitem{ref:GanjugunteGe11}
Shashidhara~K. Ganjugunte.
\newblock {\em Geometric hitting sets and their variants}.
\newblock PhD thesis, Duke University, 2011.
\newblock \url{https://dukespace.lib.duke.edu/items/391e2278-74f0-408f-9be7-4c97cf72e352}.

\bibitem{ref:Har-PeledWe12}
Sariel Har-Peled and Mira Lee.
\newblock Weighted geometric set cover problems revisited.
\newblock {\em Journal of Computational Geometry}, 3:65--85, 2012.
\newblock \href {https://doi.org/10.20382/jocg.v3i1a4} {\path{doi:10.20382/jocg.v3i1a4}}.

\bibitem{ref:KarmakarSo13}
Arindam Karmakar, Sandip Das, Subhas~C. Nandy, and Binay~K. Bhattacharya.
\newblock Some variations on constrained minimum enclosing circle problem.
\newblock {\em Journal of Combinatorial Optimization}, 25(2):176--190, 2013.
\newblock \href {https://doi.org/10.1007/s10878-012-9452-4} {\path{doi:10.1007/s10878-012-9452-4}}.

\bibitem{ref:Lev-TovPo05}
Nissan Lev-Tov and David Peleg.
\newblock Polynomial time approximation schemes for base station coverage with minimum total radii.
\newblock {\em Computer Networks}, 47:489--501, 2005.
\newblock \href {https://doi.org/10.1016/j.comnet.2004.08.012} {\path{doi:10.1016/j.comnet.2004.08.012}}.

\bibitem{ref:LiA15}
Jian Li and Yifei Jin.
\newblock A {PTAS} for the weighted unit disk cover problem.
\newblock In {\em Proceedings of the 42nd International Colloquium on Automata, Languages and Programming (ICALP)}, pages 898--909, 2015.
\newblock \href {https://doi.org/10.1007/978-3-662-47672-7_73} {\path{doi:10.1007/978-3-662-47672-7_73}}.

\bibitem{ref:LiuGe23}
Gang Liu and Haitao Wang.
\newblock Geometric hitting set for line-constrained disks.
\newblock In {\em Proceedings of the 18th Algorithms and Data Structures Symposium (WADS)}, pages 574--587, 2023.
\newblock \href {https://doi.org/10.1007/978-3-031-38906-1_38} {\path{doi:10.1007/978-3-031-38906-1_38}}.

\bibitem{ref:LiuOn23}
Gang Liu and Haitao Wang.
\newblock On the line-separable unit-disk coverage and related problems.
\newblock In {\em Proceedings of the 34th International Symposium on Algorithms and Computation (ISAAC)}, pages 51:1--51:14, 2023.
\newblock Full version available at \url{https://arxiv.org/abs/2309.03162}.

\bibitem{ref:MustafaIm10}
Nabil~H. Mustafa and Saurabh Ray.
\newblock Improved results on geometric hitting set problems.
\newblock {\em Discrete and Computational Geometry}, 44:883--895, 2010.
\newblock \href {https://doi.org/10.1007/s00454-010-9285-9} {\path{doi:10.1007/s00454-010-9285-9}}.

\bibitem{ref:PedersenOn18}
Logan Pedersen and Haitao Wang.
\newblock On the coverage of points in the plane by disks centered at a line.
\newblock In {\em Proceedings of the 30th Canadian Conference on Computational Geometry (CCCG)}, pages 158--164, 2018.
\newblock \url{https://home.cs.umanitoba.ca/~cccg2018/papers/session4A-p1.pdf}.

\bibitem{ref:PedersenAl22}
Logan Pedersen and Haitao Wang.
\newblock Algorithms for the line-constrained disk coverage and related problems.
\newblock {\em Computational Geometry: Theory and Applications}, 105-106:101883:1--18, 2022.
\newblock \href {https://doi.org/10.1016/j.comgeo.2022.101883} {\path{doi:10.1016/j.comgeo.2022.101883}}.

\bibitem{ref:WangUn23}
Haitao Wang.
\newblock Unit-disk range searching and applications.
\newblock {\em Journal of Computational Geometry}, 14:343--394, 2023.
\newblock \href {https://doi.org/10.20382/jocg.v14i1a13} {\path{doi:10.20382/jocg.v14i1a13}}.

\bibitem{ref:WangAl24}
Haitao Wang and Jie Xue.
\newblock Algorithms for halfplane coverage and related problems.
\newblock In {\em Proceedings of the 40th International Symposium on
  Computational Geometry (SoCG)}, pages 79:1--79:15, 2024.
\newblock \href {https://doi.org/10.4230/LIPIcs.SoCG.2024.79}
  {\path{doi:10.4230/LIPIcs.SoCG.2024.79}}.

\bibitem{ref:WangLi16}
Haitao Wang and Jingru Zhang.
\newblock Line-constrained $k$-median, $k$-means, and $k$-center problems in the plane.
\newblock {\em International Journal of Computational Geometry and Application}, 26:185--210, 2016.
\newblock \href {https://doi.org/10.1142/S0218195916600049} {\path{doi:10.1142/S0218195916600049}}.

\end{thebibliography}

\end{document}